\documentclass[english]{lipics-v2016}
\usepackage{xcolor,booktabs,colortbl,tabularx}

\newtheorem{proposition}[theorem]{Proposition}

\newcommand{\auth}{\textsf}
\newcommand{\book}{\textsl}

\newcommand{\jour}{\textit}
\newcommand{\vol}{}
\newcommand{\nr}[1]{(#1)}
\newcommand{\edt}[1]{(\auth{#1}, ed.)}
\newcommand{\eds}[1]{(\auth{#1}, eds.)}
\newcommand{\pub}[2]{ (#1, #2)}
\newcommand{\yr}[1]{, #1}
\newcommand{\pp}[2]{, #1--#2}

\newcommand{\suc}{\mbox{$\mathit{Suc}$}}
\newcommand{\betw}{\mbox{$\mathit{Bet}$}}
\newcommand{\thr}{\mbox{$\mathit{Th}$}}
\newcommand{\thtl}{\mbox{$\mathit{ThTL}$}}
\newcommand{\bthtl}{\mbox{$\mathit{BThTL}$}}
\newcommand{\invtl}{\mbox{$\mathit{InvTL}$}}
\newcommand{\binvtl}{\mbox{$\mathit{BInvTL}$}}
\newcommand{\utlth}{\mbox{$\mathit{\thtl}[\fut,\past]$}}
\newcommand{\btlth}{\mbox{$\mathit{\bthtl}[\fut,\past]$}}
\newcommand{\utlinv}{\mbox{$\mathit{\invtl}[\fut,\past]$}}
\newcommand{\utl}{\mbox{$\mathit{TL}[\fut,\past]$}}
\newcommand{\btlinv}{\mbox{$\mathit{\binvtl}[\fut,\past]$}}
\newcommand{\until}{\textsf{U}}
\newcommand{\since}{\textsf{S}}
\newcommand{\fut}{\textsf{F}}

\newcommand{\past}{\textsf{P}}
\newcommand{\nextt}{\textsf{X}}
\newcommand{\prev}{\textsf{Y}}
\newcommand{\gfut}[1]{\mbox{$\fut_{#1}$}}
\newcommand{\gpast}[1]{\mbox{$\past_{#1}$}}

\newcommand{\ltl}{\mbox{$\mathit{LTL}$}}
\newcommand{\ltlbin}{\mbox{$\ltl[\until,\since]$\/}}
\newcommand{\cltlbin}{\mbox{$C\ltlbin$\/}}
\newcommand{\ltlun}{\mbox{$\ltl[\fut,\past]$\/}}
\newcommand{\ltlunsuc}{\mbox{$\ltl[\fut,\past,\nextt,\prev]$\/}}
\newcommand{\ltlunsucexp}{\mbox{$\ltl[\fut,\past,\nextt^{2^n},\prev^{2^n}]$\/}}

\newcommand{\defn}{\mathrel{\mbox{$~\stackrel{\rm def}{=}~$}}}
\newcommand{\such}{.~}
\newcommand{\limplies}{\rightarrow}
\newcommand{\eqvt}{\leftrightarrow}
\newcommand{\orover}{{\displaystyle \bigvee}}
\newcommand{\andover}{{\displaystyle \bigwedge}}
\newcommand{\fo}{\mbox{$\mathit{FO}$}}
\newcommand{\fotwo}{\mbox{$\fo^2$\/}}
\newcommand{\fotwoless}{\mbox{$\fotwo[<]$\/}}
\newcommand{\fotwosuc}{\mbox{$\fotwo[<,+1]$\/}}
\newcommand{\fotwothr}{\mbox{$\fotwo[<,\thr]$\/}}
\newcommand{\fotwobet}{\mbox{$\fotwo[<,\betw]$\/}}
\newcommand{\pspace}{\mbox{$\mathit{PSPACE}$}}
\newcommand{\expspace}{\mbox{$\mathit{EXPSPACE}$}}
\newcommand{\nexptime}{\mbox{$\mathit{NEXPTIME}$}}
\newcommand{\np}{\mbox{$\mathit{NP}$}}

\begin{document}

\title{Two-variable Logic with a Between Predicate}

\author[1]{Andreas Krebs}
\affil[1]{Universit\"at T\"ubingen, krebs@informatik.uni-tuebingen.de}
\author[2]{Kamal Lodaya}
\affil[2]{The Institute of Mathematical Sciences, Chennai, kamal@imsc.res.in} 
\author[3]{Paritosh Pandya}
\affil[3]{Tata Institute for Fundamental Research, pandya@tifr.res.in}
\author[4]{Howard~Straubing}
\affil[4]{Boston College, straubin@bc.edu}

\Copyright{Andreas Krebs, Kamal Lodaya, Paritosh Pandya, Howard Straubing}
\keywords{automata theory, algebra, finite model theory, modal and temporal logics, computational complexity, verification}
\ArticleNo{p}
\DOIPrefix{}
\maketitle

\begin{abstract}
We study an extension of  $FO^2[<],$ first-order logic interpreted in finite words, in which formulas are restricted to use only two variables.  We adjoin to this language two-variable atomic formulas that say, `the letter $a$ appears between positions $x$ and $y$'.  This is, in a sense, the simplest property that is not expressible using only two variables.  

We present several logics, both first-order and temporal, that have the same expressive power, and find matching lower and upper bounds for the complexity of satisfiability for each of these formulations.  We also give an  effective necessary condition, in terms of the syntactic monoid of a regular language, for a property to be expressible in this logic. We show that this condition is also sufficient for words over a two-letter alphabet. This algebraic analysis allows us us to prove, among other things, that our new logic has strictly less expressive power than full first-order logic $FO[<].$
\end{abstract}

\section{Introduction}

We denote by $FO[<]$ first-order logic with the order relation $<$, interpreted in finite words over a finite alphabet $A.$  Variables in first-order formulas are interpreted as positions in a word, and for each letter  $a\in A$ there is a unary predicate $a(x),$ interpreted to mean `the letter in position $x$ is $a$'. Thus sentences in this logic define properties of words, or, what is the same thing, languages $L\subseteq A^*.$ The logic $FO[<]$ over words has been extensively studied, and has many equivalent characterizations in terms of temporal logic, regular languages, and the algebra of finite semigroups. (See, for instance,~\cite{Str-book,Wilke} and the many references cited therein.)

It is well known that every sentence of $FO[<]$ is equivalent to one using only three variables, but that the family of languages definable with two-variable sentences is strictly smaller \cite{IK}. The fragment $FO^2[<],$ consisting of the two-variable formulas, has also been very thoroughly investigated, and once again, there are many equivalent characterizations \cite{TW}.

The reason $FO^2[<]$ is strictly contained in $FO[<]$ is that one cannot express `betweenness' with only two variables.  More precisely, the following predicate
$$ \exists z(a(z) \wedge x < z\wedge z<y),$$
which asserts that there is an occurrence of the letter $a$ strictly between $x$ and $y,$ is not expressible using only two variables.  Let us denote this  predicate, which has two free variables,  by $a(x,y),$  and the resulting logic by \fotwobet. What properties can we express when we adjoin these new predicates to $FO^2[<]$?  The first obvious question to ask is whether we recover all of $FO[<]$ in this way.  The answer, as we shall see, is `no', but we will give a much more precise description.

The present article is a study of this extended two-variable logic.  Our investigation is centered around two quite different themes.  One theme
investigates several different
logics, based in $FO[<]$ as well as in temporal logic $LTL,$  for expressing this 
betweenness, and establishes their expressive equivalence. We explore the complexity of 
satisfiability checking in these logics as a measure of their descriptive succinctness.

\begin{table*}[t]
\setlength{\aboverulesep}{0pt}
\setlength{\belowrulesep}{0pt}
\setlength{\extrarowheight}{1ex}
 \hspace*{-0.5cm}\begin{tabularx}{15cm}{ll>{\columncolor{gray!50}}lll}
    \toprule
    Complexity/Variety & {\bf Ap} & {\bf W}  & {\bf DA$*$D } & {\bf DA} \\[1ex]
    \midrule[0.08em]
    Nonelementary      & $\fo[<]$ &  &  &  \\[1ex]
    \midrule
    \expspace          & \cltlbin & \fotwothr, & \ltlunsucexp & \\
                       &          & \fotwobet, & (binary~notation) & \\
                       &          & \bthtl, &                & \\
                       &          & \thtl &                & \\[1ex]
    \midrule
    \nexptime          &          &              & \fotwosuc      & \fotwoless \\
                       &          &              &                & (unbounded~alphabet)\\[1ex]
    \midrule
    \pspace            & \ltlbin  & \binvtl,        & \ltlunsuc   & \\
                       &          &          \invtl &             & \\[1ex]
    \midrule
    \np                &          &                &              & \fotwoless \\
                       &          &                &              & (bounded~alphabet), \\
                       &          &                &              & \ltlun \\[1ex] 
    \bottomrule
  \end{tabularx}

\caption{A summary of the results in this paper, in the context of the complexity and expressive power of temporal and predicate logics studied in earlier work. The variety {\bf W} heading the second column coincides with ${\bf M}_e{\bf DA}$ for two-generated monoids, and we conjecture that the two varieties are identical.}\label{tab:results}
\end{table*}

The second theme is devoted to determining, in a sense that we will make precise, the exact expressive power of this logic.  Here we draw on tools from the algebraic theory of semigroups.

Owing to considerations of length, we will for the most part confine ourselves to careful statements of our main results, and provide only outlines of the proofs, omitting some technical details. 

In Section~\ref{BasicProperties} we will give the precise definition of our logic
$\fotwobet$
(although there is not much more to it than what we have written in this Introduction). We introduce a related logic 
\fotwothr\/ which enforces quantitative constraints on counts of letters, and we
show that it has the same expressive power, although it can result in formulas that are considerably more succinct. 
In addition, we introduce two temporal logics,
one qualitative  and one quantitative, but again with the same expressive power as our original formulation.

In Section~\ref{KamalParitosh} we determine the complexity of formula 
satisfiability for each of these logics. 

Section~\ref{AndreasHoward} is devoted to a characterization of the expressive power of this logic in terms of the algebra of finite semigroups. This builds on earlier algebraic studies of the regular languages definable in $FO^2[<]$ \cite{TW}, and makes critical use of the algebraic theory of finite categories, as developed by Tilson~\cite{Tilson}. We find an effective necessary condition for a language to be definable in \fotwobet\/.  We conjecture that this condition is also sufficient, and prove that it is sufficient  for languages over a two-letter alphabet. One consequence is that we are able to determine effectively whether a given formula of $FO[<]$ over a two-letter alphabet is equivalent to a formula in the new logic.  We use these results to show that \fotwobet\/ is strictly less expressive than $FO[<].$ We also provide a detailed study of the quantifier alternation depth (or, what is more or less the same thing, the so-called `dot-depth') of languages definable in this logic.  

Table \ref{tab:results} gives a map of our results 
and compares them to those of previous related work. Etessami {\em et al.}~\cite{EVW} as well as Weis and Immerman  \cite{WI} have explored logics \fotwoless\/ and \fotwosuc, as well as matching temporal logics and their decision complexities. Th\'erien and Wilke~\cite{TW} found characterizations of the expressive power of these same logics, using algebraic methods. We find that our new logics are more expressive but this comes at the cost of some computational power.

Some counting extensions \cltlbin\/ of full $\ltl[\until,\since]$ have been studied by 
Laroussinie {\em et al.} \cite{LMP}, and by Alur and Henziger as discrete time Metric Temporal logic \cite{AH}.

Sketches of the proofs of the main results appear in Section~\ref{proofsketches}.

\section{ Basic Properties}\label{BasicProperties}

\subsection{Definition}
\label{sec:examples}

$FO[<]$ is first-order logic interpreted in words over a finite alphabet $A$, with a unary predicate $a(x)$ for each $a \in A,$ interpreted to mean that the letter in position $x$ is $a.$ If $\phi$ is a sentence of $FO[<],$ then the set of words $w\in A^*$ such that $w\models\phi$ is a language in $A^*,$ in fact a regular language

 For each $a\in A$ we adjoin to this logic a {\it binary} predicate $a(x,y)$ which is interpreted to mean
$$\exists z(x<z\wedge z<y\wedge a(z)).$$
This predicate cannot be defined in ordinary first-order logic over $<$ without introducing a third variable.  We will investigate the fragment \fotwobet, obtained by restricting to formulas that use  both the unary and binary $a$ predicates, along with $<,$ but use only two variables. 
\smallskip

There is an even simpler predicate that is not expressible in two-variable logic that we could have adjoined:  this is the successor relation $y=x+1.$  The logic $FO^2[<]$ supplemented by successor, which we denote by $FO^2[<,+1]$ has also been extensively studied, and the kinds of questions that we take up here for \fotwobet\/ have already been answered for $FO^2[<,+1].$ (See, for example, ~\cite{EVW,LPS,TW}).

\noindent{\bf Example.} The successor relation  $y=x+1$ is itself definable in \fotwobet, by the formula
$$x<y\wedge \bigwedge_{a\in A}\neg a(x,y).$$
As a result, we can define the set $L$  of words over $\{a,b\}$ in which there is no occurrence of two consecutive $b$'s by a sentence of \fotwobet. We can similarly define the set of words without two consecutive $a$'s.  Since we can also say that the first letter of a word is $a$ (by $\forall x(\forall y (x\leq y)\rightarrow a(x))$ and that the last letter is $b,$ we can define the language $(ab)^*$ in \fotwobet. \footnote{In contrast to the usual practice in model theory, we permit our formulas to be interpreted in the empty word:  every existentially quantified sentence is taken to be false in the empty word, and thus every universally quantified sentence is true.} This language is not, however, definable in \fotwo$[<]$.

\smallskip

\noindent{\bf Example.}  Let $L\subseteq\{a,b\}^*$ be the language defined by the regular expression
$$(a+b)^*bab^+ab(a+b)^*.$$
This language is definable in \fotwobet by the sentence
$$\exists x(\exists y(x<y\wedge b(x,y)\wedge\neg a(x,y)\wedge \alpha(y))\wedge\beta(x)),$$
where $\alpha(y)$ is
$$a(y)\land \exists x(x=y+1\wedge b(x))$$
and $\beta(x)$ is
$$a(x)\land \exists y(x=y+1\wedge b(y)).$$
As we shall see further on, this language is not definable in \fotwosuc, so our new logic has strictly more expressive power than \fotwosuc.

\smallskip

\noindent{\bf Example.}
\fotwobet\/ has the ability to represent
an $r$-bit counter (modulo $2^r$) in a word and to assert properties based
on the counter value. This is done by having successive substrings $b_0 \dots b_{r-1}$ of length $r$, with $b_0$ representing the least significant bit 
of the counter, separated by a marker (denoted here by letter $mark$).
We use two letters $0$ and $1$ for the bit values,
which the $b_i$ will range over.
The representation of an $r$-bit constant is described by the $O(r)$ size formula:
$$
mark(x) \land \suc^{r}(x)=b_0 \dots b_{r-1},
$$
where $\suc^i(x)=z$ is defined by
$$
\begin{array}{rl}
\suc^1(x)    =b \defn & \!\!\!\!\exists y \such y=x+1 \land b(y),\\
\suc^{i+1}(x)=bz \defn & \!\!\!\!\exists y \such y=x+1 \land b(y) \land \suc^{i}(y)=z, \text{~for~} i > 0.
\end{array}
$$

Clearly each specific number such as $0$, $2^r-1$ or a threshold value 
can be defined by an $O(r)$ formula.
After the number ends we will have a marker symbol again
to begin the next number. The formula
$mark(x) \land \lnot mark(x,y) \land mark(y)$
jumps from start of one number to the next one.

The $O(r^2)$ formula $EQ$ below checks equality of two numbers
by comparing the $r$ bits in succession.
We use the fact that the bit string always has $r$ bits over
the letters $\{0,1\}$ and
if we do not have $0$ where we expect a bit then we must have $1$.

\[\begin{array}{rl}
EQ(x,y)   \defn &\!\!\!\! mark(x) \land mark(y) \land \andover_{i=1}^{r} EQ_i(x,y),\\
EQ_i(x,y) \defn &\!\!\!\! \suc^i(x)=0 \eqvt \suc^i(y)=0.
\end{array}\]

By small variations of this formula, we can define formulae
$LT$, $GT$ etc, to make other comparisons.
Incrementing the counter modulo $2^r$
is encoded by an $O(r^3)$ formula $INC_1(x,y)$ 
which converts a least significant block of $1$s to $0$s,
using $r$ disjunctions of $O(r^2)$ formulas.
 \[
 \begin{array}{l}
 (\orover_i \suc^i(x)=0 \land \andover_{j<i} \suc^j(x)=1 \limplies  \\
 \hspace*{1cm}
(\suc^i(y)=1 \land \andover_{j<i} \suc^j(y)=0 \land \\
 \hspace*{1cm} \andover_{k>i} \suc^k(x)=0 \eqvt \suc^k(y)=0))\\
  \land (\andover_i \suc^i(x)=1 \limplies \andover_j \suc^j(y)=0) 
 \end{array}
\]
We can also define $INC_c(x,y)$ which checks that the number at position $y$
of the word is obtained by incrementing the number at position $x$ 
by a constant $c$.

\medskip

In contrast, it is quite difficult to find examples of languages definable in $FO[<]$ that are {\it not} definable in \fotwobet.  Much of this paper is devoted to establishing methods for generating such examples.

\subsection{Two-variable Threshold Logic}
\label{sec:twotwo}

We  generalize \fotwobet\/ as follows: Let $k\geq 0$ and $a\in A.$  We define $(a,k)(x,y)$  to mean that $x<y,$ and that there are at least $k$ occurrences of $a$ between $x$ and $y.$   Adding these (infinitely many) predicates gives a new logic $FO^2[<,Th].$  

\smallskip

\noindent{\bf Examples.} The language $STAIR_k$ consists of all words $w$ over $\{a,b,c\}$ which have a subword of the form $a (a+c)^* a $ with at least $k$
occurrences of $a$. This can be specified by sentence $\exists x \exists y ( x<y \land (a,k)(x,y) \land \neg b(x,y))$.  \\
Threshold logic is quite useful in specifying quantitative properties of systems. For example, a bus arbiter circuit may have the property that if $req$ is continuously on for $15$ cycles then there should be at least $3$ occurrences of $ack$. This can be specified by 
$\forall x \forall y ((req,15)(x,y) \rightarrow (ack,3)(x,y))$. 

\smallskip

Since $a(x,y)$ is equivalent to $(a,1)(x,y),$  $FO^2[<,Th]$ is at least as expressive as \fotwobet.  What is less obvious is that the converse is true, albeit at the cost of a large blowup in the quantifier complexity of formulas.

\begin{theorem}\label{thm.invequalsthr}
Considered as language classes,
$$\fotwobet=FO^2[<,Th].$$
\end{theorem}

There is a bit more to this than meets the eye: 
The predicates $(a,k)(x,y)$ for $k>1$ are not themselves expressible by single formulas of \fotwobet, and therefore the proof of Theorem~\ref{thm.invequalsthr} is not completely straightforward.

\subsection{Temporal Logic}

We denote by $\utl$ temporal logic with two operators $\fut$ and $\past.$  Atomic formulas are the letters $a\in A.$  Formulas are built from atomic formulas by applying the boolean operations $\wedge,\vee,$ and $\neg,$ and the modal operators $\phi\mapsto \fut\phi,$ $\phi\mapsto\past\phi.$  

We interpret these formulas in {\it marked words} $(w,i),$ where $w\in A^*$ and $1\leq i\leq |w|.$  Thus $(w,i)\models a$ if $w(i)=a,$ where $w(i)$ denotes the $i^{th}$ letter of $w.$ Boolean operations have the usual meaning. We define $(w,i)\models \fut\phi$ if there is some $j>i$ such that $(w,j)\models\phi,$ and $(w,i)\models\past\phi$ if there is some $j<i$ with $(w,j)\models\phi.$

We can also interpret a formula in ordinary, that is, unmarked words, by defining $w\models\phi$ to mean $(w,1)\models\phi.$  Thus temporal formulas, like first-order sentences, define languages in $A^*.$  The temporal logic $\utl$ is known to define exactly the languages definable in $FO^2[<]$ \cite{EVW,TW}.

We now define new temporal logics by modifying the modal operators $\fut$ and $\past$ with {\it threshold constraints}---these are versions of the between predicates $a(x,y)$ and $(a,k)(x,y)$ that we introduced earlier.   Let $B\subseteq A.$  A threshold constraint is an expression of the form $\#B\sim c,$  where $c\geq 0,$ and $\sim$ is one of the symbols $\{<,\leq,>,\geq,=\}.$ Let  $w\in A^*$ and $1\leq i<j\leq |w|.$ We say that $(w,i,j)$ satisfies the threshold constraint $\#B\sim c$ if
$$|\{k: i<k<j\text{ and }w(k)\in B\}|\sim c.$$
We can combine threshold constraints with boolean operations $\wedge,\vee,\neg.$  We define satisfaction of a boolean combination of threshold constraints in the obvious way--that is, $w(i,j)$ satisfies $g_1\vee g_2$ if and only if $(w,i,j)$ satisfies  $g_1$ or $g_2,$ and likewise for the other boolean operations.

If $g$ is a boolean combination of threshold constraints, then our new operators $\gfut{g}$ and $\gpast{g}$ are defined as follows: $(w,i)\models \gfut{g}\phi$ if there exists $j>i$ such that $(w,i,j)$ satisfies $g$ and $(w,j)\models\phi$, $(w,i)\models\gpast{g}\phi$ if and only if there exists $j<i$ such that $(w,j,i)$ satisfies $g$ and $(w,j)\models\phi$.

\smallskip
\noindent{\bf Examples.}
\smallskip

\noindent We can express $\fut\phi$ with threshold constraints as $\gfut{\#\emptyset=0}\phi.$

\smallskip

\noindent We use $\nextt$ to denote the `next' operator:  $(w,i)\models \nextt\phi$ if and only if $(w,i+1)\models\phi.$ We can express this with threshold constraints by $\gfut{\#A=0}\phi.$

\smallskip

\noindent We can define the language $(ab)^+$ over the alphabet $\{a,b\}$ as the conjunction of several subformulas: $a\wedge \nextt b$ says that the first letter is $a$ and the second $b.$  $\neg \fut (a\wedge \nextt  a)$ says that no  occurrence of $a$ after the first letter is immediately followed by another $a,$ and similarly we can say that no occurrence of $b$ is followed immediately by another $b.$ The formula $\fut(b\wedge\neg\nextt (a\vee b))$ says that the last letter is $b.$

\smallskip

\noindent 
It is useful to have boolean combinations of threshold constraints.
The language $STAIR_k$ given in Section \ref{sec:twotwo} can be defined by 
$\fut (\gfut{\#a=k \land \#b=0} ~true)$.

\smallskip

We denote by $\btlth$ temporal logic with these modified operators $\gfut{g}$ and $\gpast{g},$ where $g$ is a boolean combination of threshold constraints.  We also define several fragments of  $\btlth$:  In $\utlth$ we restrict the constraints $g$ to be atomic threshold constraints, rather than boolean combinations.  In $\utlinv$ we restrict to constraints of the form $\#B=0$--we call these {\it invariant constraints} and in $\btlinv$ to boolean combinations of such constraints.

\begin{theorem}~\label{thm.tlequivalence} The logics $\utlth,$ $\btlth,$ $\utlinv,$ $\btlinv,$ $FO^2[<,Th],$ and \fotwobet\/ all define the same family of languages.
\end{theorem}

\section{Complexity of Satisfiability}\label{KamalParitosh}

Given a formula in one of these logics, what is the computational complexity of determining whether it has a model, that is, whether the language it defines is empty or not?  This is the {\it satisfiability problem}  for the logic.  To determine this, we require some way to measure the size of the input formula.  For formulas containing threshold constraints, we code the threshold value in binary, so that mention of a threshold constant $c$ contributes $\lceil\log_2 c\rceil$ to the size of the formula. Mention of a subalphabet $B$ contributes $|B|$ to the size of the formula.

In the verification literature (\cite{EVW} is relevant for
this paper) the syntax allows 
a finite set of propositional letters $PV$
which may or may not simultaneously  hold at a position of a word. 
This allows us to compactly talk about large alphabets.
One can think of the alphabet $A$ as the set of valuations $2^{PV}$ 
to get finite word models over $A$.
Thus the alphabet $A$ is given a boolean algebra structure and subsets of $A$
are specified using propositions over $PV$. 
Our results below hold for bounded and unbounded alphabets, 
which may be explicitly specified or symbolically specified by propositions.

\begin{theorem}\label{thm.utlsat}
Satisfiability of the temporal logics $\invtl$ (with invariant constraints) and
$\bthtl$ (with threshold constraints) 
is complete for \pspace\/ and \expspace, respectively.
\end{theorem}

\begin{theorem}
\label{thm.fo2sat}
Satisfiability of the two-variable logics
$\fotwobet$ and $\fotwothr$ is \expspace-complete.
\end{theorem}

\section{Algebraic Characterization}\label{AndreasHoward}

\subsection{Background on finite monoids and varieties}

For further background on the basic algebraic notions in this section, see Pin~\cite{Pin}. 

A {\it monoid} is a set together with an associative multiplication (that is, it is a {\it semigroup}) and a multiplicative identity 1.

All of the languages defined by sentences of $FO[<]$ are regular languages. Our characterization of languages in \fotwobet\/ is based on properties of the {\it syntactic monoid} $M(L)$ of a regular language $L.$  This is the transition monoid of the minimal deterministic automaton recognizing $L,$ and therefore a finite monoid. Equivalently, $M(L)$ is  the smallest monoid $M$ that {\it recognizes} $L$ in the following sense:  There is a homomorphism $\phi:A^*\to M$ and a subset $X\subseteq M$ such that $L=\phi^{-1}(X).$

Let $M$ be a finite monoid. An {\it idempotent} $e\in M$ is an element satisfying $e^2=e.$  If $m\in M,$ then there is some $k\geq 1$ such that $m^k$ is idempotent.  This idempotent power of $m$ is unique, and we denote it by $m^{\omega}.$ 

A finite monoid is {\it aperiodic} if it contains no nontrivial groups, equivalently, if it satisfies the identity $x\cdot x^{\omega}=x^{\omega}$ for all $x\in M.$  We denote the class of aperiodic finite monoids by ${\bf Ap}.$  ${\bf Ap}$ is a {\it variety} of finite monoids: this means that it is closed under finite direct products, submonoids, and quotients.

A well-known theorem, an amalgam of results of McNaughton and Papert~\cite{MNP} and of Sch\"utzenberger~\cite{Sch1}, states that $L\subseteq A^*$ is definable in $FO[<]$ if and only if $M(L)\in{\bf Ap}.$  This situation is typical:  Under very general conditions, the languages definable in fragments of $FO[<]$ can be characterized as those whose syntactic monoids belong to a particular variety ${\bf V}$ of finite monoids.(See Straubing~\cite{Str-latin}.)

If $M$ is a finite monoid and $m_1,m_2\in M,$ we write $m_1\leq_{\cal J} m_2$ if $m_1=sm_2t$ for some $s,t\in M.$  This is a preorder, the so-called ${\cal J}$-ordering on $M.$ If $e\in M$ is idempotent, then we denote by $M_e$ the submonoid of $M$ generated by elements $m$ such that $e\leq_{\cal J} m.$ Observe that $eM_ee$ is a subsemigroup of $M_e,$ in fact a monoid whose identity element is $e.$

If ${\bf V}$ is a variety of finite monoids, then we can form a new variety ${\bf M}_e{\bf V}$ as follows:
$${\bf M}_e{\bf V}=\{M:eM_ee\in {\bf V}\text{ for all $e^2=e\in M$}\}.$$

\begin{proposition}~\label{prop.emee}
${\bf M}_e{\bf V}$ is a variety of finite monoids.
\end{proposition}

Let ${\bf I}$ denote the variety consisting of the trivial one-element monoid alone. We define the class 
$${\bf DA}={\bf M}_e{\bf I}.$$
That is, ${\bf DA}$ consists of those finite monoids $M$ for which $eM_ee=e$ for all idempotents $e\in M.$ By Proposition~\ref{prop.emee}, ${\bf DA}$ is a variety of finite monoids. The variety ${\bf DA}$ was introduced by Sch\"utzenberger~\cite{Sch2} and it figures importantly in work on two-variable logic.   Th\'erien and Wilke showed that a language $L$ is definable in $FO^2[<]$ if and only if $M(L)\in{\bf DA}$~\cite{TW}.

\smallskip
\noindent{\bf Example.}  Consider the language $L\subseteq \{a,b\}^*$ consisting of all words whose first and last letters are the same. The syntactic monoid of $L$ contains five elements
$$M(L)=\{1,(a,a),(a,b),(b,a),(b,b)\},$$
with multiplication given by $(c,d)(c',d')=(c,d'),$ for all $c,d,c',d'\in\{a,b\}.$  Observe that every element of $M(L)$ is idempotent.  For every $e\neq 1,$ $eM(L)e=e,$ and if $e=1,$ then $M(L)_e=1.$  Thus, $M(L)\in{\bf DA}.$  The logical characterization then tells us that $L$ is defined by a sentence of $FO^2[<].$  Indeed, $L$ is defined by
$$\exists x(\forall y(x\leq y)\wedge\exists y(\forall x(x\leq y)\wedge(a(x)\leftrightarrow a(y)))).$$

\noindent{\bf Example.} Consider the language $(ab)^*$.  We claimed earlier that it is not definable in $FO^2[<].$ We can prove this using the algebraic characterization of the logic.  The elements of the syntactic  monoid $M$ are 
$$1,a,b,ab,ba,0.$$
The multiplication is determined by the rules $aba=a,$ $bab=b,$ and $a^2=b^2=0.$ Then $ab$ and $ba$ are idempotents, and $M_{ab}=M_{ba} = M.$ Thus $ab\cdot M_{ab}\cdot ab=\{ab,0\},$ which shows that $M\notin {\bf DA},$ and thus $(ab)^*$ is not definable in $FO^2[<].$.

\smallskip

\noindent{\bf Example.}  Now consider the language given by the regular expression $(a+b)^*bab^+ab(a+b)^*.$  We saw earlier that it is definable in \fotwobet, and claimed that it could not be defined in \fotwosuc.  Th\'erien and Wilke~\cite{TW} also give an algebraic characterization of \fotwosuc :  Let $S$ be subsemigroup of the syntactic monoid of $L$ generated by nonempty words (the {\it syntactic semigroup} of $L\cap A^+$).  $L$ is definable in \fotwosuc\/ if and only if for each idempotent $e\in S,$ the monoid $N=eSe$ is in {\bf DA}.  For the language under discussion, let us denote the image of a word $w$ in the syntactic monoid by $\overline{w}.$  Then $e=\overline{b}$ is idempotent, and $f=\overline{baab}$ is an idempotent in $N=eSe.$ Let $s=\overline{bab}.$  Then $s\in eSe,$ and $fsf=f,$ so $s\in M_f.$  
We now have
$fsf=f\neq fssf,$ since $fssf=\overline{babab}$ is the zero of $N.$ Thus $fM_ff$ contains more than one element, so $eSe\notin {\bf DA}.$ Consequently this language, while definable in \fotwobet, cannot be defined in \fotwosuc.

\subsection{The main result}

\begin{theorem}\label{thm.maintheorem}
Let $L\subseteq A^*.$
If $L$ is definable in \fotwobet\/ then $M(L)\in {\bf M}_e{\bf DA}.$ Further, if $|A|=2,$ and $M(L)\in {\bf M}_e{\bf DA},$ then $L$ is definable in \fotwobet\/.
\end{theorem}

 We conjecture that sufficiency of the condition holds for all finite alphabets, not just those with two letters.  We can prove from rather abstract principles that there is some variety ${\bf W}$ of finite monoids that characterizes definability in \fotwobet\/ in this way.  Our results imply that ${\bf W}$ coincides with ${\bf M}_e{\bf DA}$ for monoids generated by two elements. The theorem provides an effective method for determining whether a given language over $A=\{a,b\}$ (given, say, by a finite automaton that recognizes it, or by a regular expression) is definable in \fotwobet, since we can compute the multiplication table of the syntactic monoid, and check whether it belongs to ${\bf M}_e{\bf DA}.$

\smallskip
\noindent {\bf Example.}  In our example above, where $L=(ab)^*,$ all the submonoids $eM_ee$ are either trivial, or are two-element monoids isomorphic to $\{0,1\},$ which is in {\bf DA}.  Thus $(ab)^*$ is definable in \fotwobet, as we saw earlier by construction of a defining formula.

\smallskip

The following corollary to our main theorem answers our original question of whether $FO[<]$ has strictly more expressive power than \fotwobet.
To prove it, we need only calculate the syntactic monoid of the given language, and verify that it is in {\bf Ap} but not in ${\bf M}_e{\bf DA}.$

\begin{corollary}~\label{cor.strictinclusion} The language given by the regular expression
$$(a(ab)^*b)^*$$
is definable in $FO[<]$ but not in \fotwobet.

\end{corollary}

\subsection{Alternation depth}

We are  interested in how \fotwobet\/ sits inside $FO[<].$  One way to measure the complexity of a language in $FO[<]$ is by the smallest number of alternations of quantifiers required in a defining formula, that is the smallest $k$ such that the language is definable by a boolean combination of sentences of $\Sigma_k[<].$  We will call this the {\it alternation depth} of the language.  (This is closely related to the {\it dot-depth}, which can be defined the same way, but with slightly different base of atomic formulas.)

\begin{theorem}\label{thm.alternationdepth}

 The alternation depth of languages in \fotwobet\/ is unbounded.
 If $|A|=2,$ then the alternation depth of languages in $A^*$ definable in \fotwobet\/ is bounded above by 3.
\end{theorem}

For alphabets of more than two letters, we conjecture that the alternation depth is also bounded by a linear function of $|A|.$

We stress that the alternation depth is measured with respect to arbitrary first-order sentences, not the variable-restricted sentences of \fotwobet.

\section{Outlines of the proofs}\label{proofsketches}

\subsection{A game characterization of $\pmb{\fotwobet}$.}

We write $(w_1,i_1)\equiv_k (w_2,i_2)$ if these two marked words satisfy exactly the same formulas of \fotwobet\  with one free variable of quantifier depth no more than $k.$

We overload this notation, and also write $w_1\equiv_k w_2$ if $w_1$ and $w_2$ are ordinary words that satisfy exactly the same {\it sentences} of \fotwobet\/ with quantifier depth $\leq k.$

Let $k\geq 0.$ The game is played for $k$ rounds in two marked words $(w_1,i_1)$ and $ (w_2,i_2)$ with a single pebble on each word.  At the start of the game, the pebbles are on the marks $i_1$ and $i_2.$   In each round, the pebble is moved to a new position in both words, producing two new marked words.

Suppose that at the beginning of a round, the marked words are $(w_1,j_1)$ and $(w_2,j_2).$  Player 1 selects one of the two words and moves the pebble to a different position.  Let's say he picks $w_1,$ and moves the pebble to $j_1',$ with $j_1\neq j_1'.$  Player 2 moves the pebble  to a new position $j_2'$ in $w_2.$  This response is required to satisfy the following properties:
\begin{enumerate}
\renewcommand\labelenumi{(\roman{enumi})}
\item The moves are in the same direction:  $j_1<j_1'$ iff $j_2<j_2'.$
\item The letters in the destination positions are the same: $w_1(j_1')=w_2(j_2').$
\item The set of letters jumped over is the same---that is, assuming $j_1<j_1'$:
\begin{eqnarray*}
&\hspace{-5mm}&\{a\in A: w_1(k)=a\text{ for some } j_1<k<j_1'\}=\\
&\hspace{-5mm}&\{a\in A: w_2(k)=a\text{ for some } j_2<k<j_2'\}.
\end{eqnarray*}
\end{enumerate}

Player 2 wins the 0-round game if $w_1(i_1)=w_2(i_2).$ Otherwise, Player 1 wins the 0-round game.

Player 2 wins the $k$-round game for $k>0$ if she makes a legal response in each of $k$ successive rounds, otherwise Player 1 wins.

The following theorem is the standard result about Ehrenfeucht-Fra\"{\i}ss\'e games adapted to this logic.

\begin{theorem}\label{thm.game_equiv}
$(w_1,i_1)\equiv_k(w_2,i_2)$ if and only if Player 2 has a winning strategy in the $k$-round game in the two marked words.
\end{theorem}

We now define the $k$-round game in ordinary unmarked words $w_1,w_2\in A^*.$  Player 1 begins in the first round by placing a pebble on a position in one of the two words, and Player 2 must respond on a position in the other word containing the same letter.  Thereafter, they play the game in the two marked words that result for $k-1$ rounds.  The following is a direct consequence of the preceding theorem.

\begin{corollary}\label{cor.game_equiv}  Player 2 has a winning strategy in the $k$-round game in $w_1$ and $w_2$ if and only if $w_1\equiv_k w_2.$
\end{corollary}

\subsection{Proof of Theorem~\ref{thm.invequalsthr}}

We introduce a game characterizing $FO^2[<,Th].$ Let $\theta$ be a function from $A$ to the positive integers. We consider formulas in $FO^2[<,Th]$ in which for all $a\in A,$ every occurrence of the predicate $(a,k)(x,y)$ has $k\leq\theta(a).$  Let's call these {\it $\theta$-bounded} formulas.

The rules of the game are the same as those for the \fotwobet\/ game, with this difference:  At each move, for each $a\in A,$ the number $m_1$ of $a$'s jumped by Player 1 must be equivalent, threshold $\theta(a),$ to the number $m_2$ of $a$'s jumped by Player 2. That is, either $m_1$ and $m_2$ are both greater than or equal to $\theta(a),$ or $m_1=m_2.$  Observe that the game for \fotwobet\/ is the case $\theta(a)=1$ for all $a\in A.$

Let us define, for marked words $(w_1,i_1),(w_2,i_2),$ $(w_1,i_1)\equiv^{\theta}_k(w_2,i_2)$ if and only if they satisfy exactly the same $\theta$-bounded formulas of quantifier depth less than or equal to $k.$  As with the case of \fotwobet, we also have a version of both the game and the equivalence relation for ordinary words.

It is easy to show that the analogues of Theorem~\ref{thm.game_equiv} and Corollary~\ref{cor.game_equiv} hold in this more general setting: Player 2 has a winning strategy in the $k$-round game in $(w_1,i_1),(w_2,i_2)$ if and only if $(w_1,i_1)\equiv^{\theta}_k(w_2,i_2),$ and likewise for ordinary words.

Let $\theta, \theta'$ be two functions from $A$ to the natural numbers that differ by one in the following sense:  $\theta'(b)=\theta(b)+1$ for exactly one $b\in A,$ and $\theta'(a)=\theta(a)$ for all $a\neq b.$  Obviously $\equiv^{\theta'}_k$ is finer than $\equiv^{\theta}_k.$  What we will show precisely is this:  $\equiv^{\theta}_{2k}$ refines $\equiv^{\theta'}_{k}.$ We prove refinement by a simple game argument, showing that if Player 2 has a winning strategy in the $2k$-move $\theta$ game in $(w_1,w_2)$ then she has a strategy in the $k$-move $\theta'$ game in the same two words. 

This will give  the desired result, because any threshold function $\theta$ can be built from the base threshold function that assigns 1 to each letter of the alphabet by a sequence of steps in which we add 1 to the threshold of each letter.  So it follows by induction that for any $\theta,$ $\equiv^{\theta}_k$ is refined by $\equiv_{k\cdot 2^r},$ where $r=\sum_{a\in A}\theta(a)-1.$  In particular, each $\equiv^{\theta}_k$-class is definable by a \fotwobet\/ sentence, although the quantifier depth of this sentence is exponential in the thresholds used.

So given a fixed $FO^2[<,Th]$ sentence $\phi$ there is a threshold such that the sentence is $\theta$-bounded. Let $k$ be the quantifier depth of $\phi$, then the sentence cannot distinguish words that are in the same equivalence class with respect to $\equiv^{\theta}_k$. As there are only finitely many $\theta$-bounded sentence of quantifier depth at most $k$, there are only finitely many such equivalence classes. By the argument above we can find a \fotwobet\ sentence for each equivalence class accepted by $\phi$ and the disjunction of these will be a sentence in \fotwobet\ that accepts the same models as $\phi$.

Observe that this argument applies to both ordinary words and marked words, and thus each formula of $FO^2[<,Th]$ with one free variable can similarly be replaced by an equivalent formula of \fotwobet.

\subsection{Proof of Theorem~\ref{thm.tlequivalence}} 

In terms of language classes, we obviously have 
$$\utlinv\subseteq\btlinv\subseteq\btlth,$$
and
$$\utlinv\subseteq\utlth\subseteq\btlth.$$

So it is enough to show that $\btlth\subseteq\utlinv.$ In performing these translations, we will only discuss the future modalities, since the past modalities can be treated the same way. 

We can directly translate any formula in $\btlth$ into an equivalent formula of $FO^2[<,Th]$ with a single free variable $x$:  A formula $\gfut{g}\psi,$ where $g$ is a boolean combination of threshold constraints, is replaced by a quantified formula $\exists y (y>x \wedge \alpha\wedge \beta(y))$, where $\alpha$ is a boolean combination of formulas $(a,k)(x,y)$ and $\beta$ is the translation of $\psi.$

We know from Theorem~\ref{thm.invequalsthr} that any formula of $FO^2[<,Th]$ can in turn be translated into an equivalent formula of \fotwobet. Furthermore, \fotwobet\/ is  equivalent to $\btlinv$: A simple game argument shows that equivalence of marked words with respect to $\btlinv$ formulas of modal depth $k$ is precisely the relation $\equiv_k$ of equivalence with respect to \fotwobet\/ formulas with quantifier depth $k.$ 

So it remains to show that $\btlinv\subseteq \utlinv$:  We do this by translating $\gfut{g}\phi$, where $g$ is a boolean combination of constraints of the form $\#B=0,$ into a formula that uses only single constraints of this form.  We can rewrite the boolean combination as the disjunction of conjunctions of constraints of the form $\#\{a\}=0$ and $\#\{a\}>0.$  Since, easily, $\gfut{g_1\vee g_2}\phi$ is equivalent to $\gfut{g_1}\phi \vee\gfut{g_2}\phi,$ we need only treat the case where $g$ is a conjunction of such constraints. We  illustrate the  general procedure  for translating such conjunctions with an example.  Suppose $A$ includes the letters $a,b,c.$  How do we express $\gfut{g}\phi,$ where $g$ is $(\#\{a\}=0)\wedge(\#\{b\}>0)\wedge(\#\{c\}>0)$? The letters $b$ and $c$ must appear in the interval between the current position and the position where $\phi$ holds.  Suppose that $b$ appears before $c$ does.  We write this as
$$\gfut{\#\{a,b,c\}=0}(b\wedge\gfut{\#\{a,c\}=0}(c\wedge\gfut{\#\{a\}=0}\phi)).$$
We take the disjunction of this with the same formula in which the roles of $b$ and $c$ are reversed.

\subsection{Proof of Theorem~\ref{thm.utlsat}}

Satisfiability of \utlinv\/ is \pspace-hard \cite{SC} since it includes
\ltlunsuc.
We observe that \btlinv\/ can be translated into \ltlbin\/ in polynomial time,
hence its satisfiability is \pspace-complete \cite{SC}.
Using a threshold constant $2^n$, written in binary in the formula with size $n$,
the  $2^n$-iterated Next operator $X^{(2^n)}$ can be expressed in \utlth,
so we obtain that its satisfiability is \expspace-hard \cite{AH}.
By an exponential translation of \btlth\/ into \btlinv
(the proof of Theorem~\ref{thm.invequalsthr} shows that such translation exists), or alternately by a polynomial 
translation into \cltlbin \cite{LMP}),
its satisfiability is \expspace-complete\/.

\subsection{Proof of Theorem~\ref{thm.fo2sat}}

We begin by reducing the exponential Corridor Tiling problem 
to satisfiability of \fotwobet. 
It is well known that this problem is \expspace-complete \cite{Furer}.

\paragraph{The exponential Corridor Tiling problem:} 
An instance $M$ is given by $(T,H,V,s,f,n)$
where $T$ is a finite set of tile types with $s,f \in T$,
the horizontal and vertical tiling relations
$H, V \subseteq T \times T$, and $n$ is a natural number.
A solution of the $2^n$ sized corridor tiling problem is
a natural number $m$ and map $\pi$ from
the grid of points $\{ (i,j) ~\mid~ 0 \leq i \leq 2^n,~0 \leq j < m\}$
to $T$ such that: \\
$\pi(0,0)=s$, $\pi(n-1,m-1)=f$ and for all $i$, $j$ on the grid, \\ 
$(\pi(i,j),\pi(i,j+1)) \in V$ and $(\pi(i,j),\pi(i+1,j)) \in H$.

\begin{lemma}
\label{lem.tiling}
Satisfiability of the two-variable logic \fotwobet\/ is \expspace-hard.
\end{lemma}
\begin{proof}
Given an instance $M$ as above of a Corridor Tiling problem,
we encode it as a sentence  $\phi(M)$ of size $poly(n)$ with a
modulo $2^n$ counter $C(x)$ encoded serially with $n+1$ letters 
as in the example in Section \ref{sec:examples}.
The marker now represents a tile and a colour from $red$, $blue$ and $green$ 
(requiring subalphabet size $3|T|$).
Thus the $2^{n} \times m$ tiling is represented by a word of length
$m(n+1)2^n$ over an alphabet of size $3|T|+2$.
The claim is that $M$ has a solution iff $\phi(M)$ is satisfiable.

The sentence  $\phi(M) \in \fotwobet$ is a conjunction of 
the following properties.
The key idea is to cyclically use monadic predicates 
$red(x), green(x), blue(x)$ for assigning colours to rows.

\begin{itemize}
\item
Each marker position has exactly one tile and one colour.
\item
The starting tile is $s$, the initial colour is $red$, 
the initial counter bits read $0^n$, the last tile is $f$.
\item
Tile colour remains same in a row and it cycles in order 
$red,green,blue$ on row change.
\item
The counter increments modulo $2^n$ in consecutive positions.  
\item
For horizontal compatibility we check:\\
$
\begin{array}{l}
\forall x \forall y \such mark(x) \land mark(y) \land \neg mark(x,y)
\limplies \\
 \hspace*{1cm} \orover_{(t_1,t_2) \in H} ~ t_1(x) \land t_2(y) 
\end{array}
$.
\item
For vertical compatibility, %
we check that $x,y$ are in adjacent rows by invariance of lack of one colour.
We check that $x$ and $y$ are in the same column by checking that 
the counter value (which encodes column number) is the same: \\
$ \forall x \forall y \such  (x < y) \land 
(\neg red(x,y) \lor \neg blue(x,y) \lor \neg green(x,y)) \land EQ(x,y) \limplies
\orover_{(t_1,t_2) \in V} ~ t_1(x) \land t_2(y)$.
\end{itemize}

It is easy to see that we can effectively translate an instance $M$ of 
the exponential corridor tiling problem into $\phi(M)$ in time 
polynomial in $n$. The translation preserves satisfiability. 
Hence, by reduction, satisfiability of $\fotwobet$ over bounded
as well as unbounded alphabets is $\expspace$-hard.
\end{proof}

\begin{lemma}
\label{lem.reduction}
There is a satisfiability-preserving polynomial time reduction from 
the logic \fotwothr\/ to the logic \fotwobet.
\end{lemma}
\begin{proof}
We give a polytime reduction from \fotwothr\/ to \linebreak \fotwobet\/ 
which preserves satisfiability. 
We consider in the extended syntax
a threshold constraint  $\#a(x,y)=k$ where $a$ is a letter or a proposition and
$k$ is a natural number.
The key idea of the reduction is illustrated by the following example.

For each threshold constraint $g$ of the form $\#a(x,y)=2^r$,
we specify a \emph{global} modulo $2^r$ counter $C_g$ 
using monadic predicates
$p_1(x), \ldots, p_r(x)$. (This requires a symbolic alphabet 
where several such predicates may be true at the same position.) 
By ``global'' we mean that the counter $C_g$ has value $0$ 
at the beginning of the word and it increments whenever $a(x)$ is true. 
This is achieved using the formula:
\[
\forall x,y \such y=x+1 \limplies (a(x) \limplies INC_g(x,y)) \land
 (\neg a(x) \limplies EQ_g(x,y))
\]
Also, we have three colour predicates $red_g(x)$, $blue_g(x)$ and $green_g(x)$
where the colour at the beginning of the word is $red_g$,
and we change the colour cyclically each time the counter $C_g$
resets to zero by overflowing.
As in the proof of Lemma~\ref{lem.tiling},
invariance of lack of one colour and the fact that $x,y$ have different colours
ensures that the counter overflows at most once. 
We replace the constraint $g$
of the form $\#a(x,y)=2^r$ by an equisatisfiable formula:
\[
\begin{array}{l}
x<y \land EQ_g(x,y) \land \#a(x,y) > 0 \land\\
(\neg red_g(x,y) \lor \neg blue_g(x,y) \lor \neg green_g(x,y))
\end{array}
\]
More generally, we define a polynomial sized quantifier free formula
$INC_{g,c}(x,y)$ for any given constant $c$ with $2^{r-1} <c \leq 2^r$ using
propositions $p_1, \ldots, p_r$ and three colour predicates.
The formula asserts that $\#a(x,y)+ c = 2^r$. 
Using this we can encode the constraint $\#a(x,y)=2^r-c$ for any $c$. 
Similarly to $EQ_g(x,y)$, we can also
define formulae $LT_g(x,y)$ to denote that its counter $C_g(x) < C_g(y)$,
$GT_g(x,y)$ to denote that $C_g(x) > C_g(y)$, etc.
Hence any form of threshold counting
predicate can be replaced by an equisatisfiable formula, with
all these global counters running from the beginning of the word to the end.
Thus we have a polynomially sized equisatisfiable reduction from
\fotwothr\/ to \fotwobet. 
\end{proof}

To complete the proof of Theorem~\ref{thm.fo2sat},
the upper bound for \fotwothr\/ comes from
Lemma \ref{lem.reduction}, an exponential translation
from \fotwobet\/ to \btlinv\/ using an order type argument similar to \cite{EVW}
(our Theorem~\ref{thm.tlequivalence} also points to this equivalence),
and the \pspace\/ upper bound for \btlinv\/ (Theorem~\ref{thm.utlsat}).

\subsection{Proof of necessity in Theorem~\ref{thm.maintheorem}}

Here we outline the proof of the direction of Theorem~\ref{thm.maintheorem} stating that every language definable in \fotwobet\/ has its syntactic monoid in ${\bf M}_e{\bf DA}.$  This is all we will need to prove Corollary~\ref{cor.strictinclusion} and the first assertion of Theorem~\ref{thm.alternationdepth}.

We first use a game argument to prove the following fact: 

\begin{lemma}\label{lem.morphismclosure}
Lef $k\geq 0,$ $A,B$ finite alphabets, and $f:B^*\to A^*$ a monoid homomorphism.  Let $w_1,w_2\in B^*.$  If  $w_1\equiv_k w_2,$ then $f(w_1)\equiv_k f(w_2).$
\end{lemma}

Now let $L\subseteq A^*$ be definable by a sentence of \fotwobet.   $L$ is a union of $\equiv_k$-classes for some $k,$ and thus $L$ is recognized by the quotient monoid $N=A^*/\equiv_k,$ so $M(L)$ is a homomorphic image of this monoid.  Consequently, it is sufficient to show that $N$ itself is a member of the variety ${\bf M}_e{\bf DA}.$  We denote by $\psi:A^*\to N$ the projection morphism onto this quotient.

Take $e=e^2\in N$ and $x,y\in eN_ee.$ We will show
$$(xy)^{\omega}x(xy)^{\omega}=(xy)^{\omega}.$$
This identity characterizes the variety {\bf DA} (see, for example Diekert, {\it et al.}~\cite{DGK}), so this will prove $N\in {\bf M}_e{\bf DA},$ as required. We can write
$$x=em_1\cdots m_re, y=em_1'\cdots m_s'e,$$
where each $m_i$ and $m_j'$ is $\leq_{\cal J}$-above $e.$ Thus we have
$$e=n_{2i-1}m_in_{2i}=n_{2r+2j-1}m_j'n_{2r+2j},$$
for some $n_1,\ldots n_{2r+2s}\in N.$

Now let $B$ be the alphabet
$$\{a_1,\ldots,a_r,b_1,\ldots,b_s,c_1,\ldots,c_{2r+2s}\},$$
  We define a homomorphism $\phi:B^*\to N$ by mapping each $a_i$ to $m_i,$ $b_j$ to $m_j',$ and $c_k$ to $n_k.$  Since $\psi:A^*\to N$ is onto, we also have a homomorphism $f:B^*\to A^*$ satisfying $\psi\circ f=\phi.$ 
We define words $v,{\bf a},{\bf b}, X_{S,T}\in B^*,$ where $S,T>0,$ as follows:
$$v=\prod_{i=1}^rc_{2i-1}a_ic_{2i}\cdot\prod_{i=r+1}^sc_{2i-1}b_{i-r}c_{2i},$$
$${\bf a}=a_1\cdots a_r, {\bf b}=b_1\cdots b_r,$$
$$X_{S,T}=(v^S{\bf a}v^S{\bf b}v^S)^T.$$

 Observe that
$$\phi(v)=e, \phi(X_{S,T})=(xy)^T,\phi({\bf a})=m_1\cdots m_r,\phi({\bf b})=m'_1\cdots m'_s.$$
Since all languages recognized by $N$ are definable in $FO[<],$ $N$ is aperiodic, and thus for sufficiently large values of $T,$ we have $\phi(X_{S,T})=(xy)^{\omega}.$  Thus for sufficiently large $T$ and all $S$ we have
$$\phi(X_{S,T}X_{S,T})=(xy)^{\omega}(xy)^{\omega}=(xy)^{\omega},$$
 $$\phi(X_{S,T}{\bf a}X_{S,T})=(xy)^{\omega}x(xy)^{\omega}.$$

By Lemma~\ref{lemma.xst} (proved next),

 $$X_{S,T}X_{S,T}\equiv_k X_{S,T}{\bf a}X_{S,T}.$$

Then by Lemma~\ref{lem.morphismclosure},

$$f(X_{S,T}X_{S,T})\equiv_k f(X_{S,T}{\bf a}X_{S,T}).$$
So
\begin{eqnarray*}
(xy)^{\omega} &=&\phi(X_{S,T}X_{S,T})\\
&=& \psi(f(X_{S,T}X_{S,T}))\\
&=& \psi(f(X_{S,T}{\bf a}X_{S,T}))\\
&=& \phi(X_{S,T}{\bf a}X_{S,T})\\
&=& (xy)^{\omega}x(xy)^{\omega},
\end{eqnarray*}
as claimed.

\subsection{Proof of Lemma~\ref{lemma.xst}}

The essential idea is to show that a certain identity is satisfied defining this variety of finite monoids, which we do by application of our game characterization of $FO^2[<,Inv].$ 

We begin by considering an alphabet $B$ of the form
$$B=\{a_1,\ldots,a_r,b_1,\ldots,b_s,c_1,\ldots,c_t\},$$
where $t=2r+2s.$
Let $v\in B^*$ be the word
$$c_1a_1c2\cdot c_3a_2c_4\cdots c_{2r-1}a_rc_{2r}\cdot c_{2r+1}b_1c_{2r+2}\cdots c_{2r+2s-1}b_sc_{2r+2s}.$$
 Let $R>0.$ We will build words by concatenating the factors
$${\bf a}=a_1\cdots a_r, {\bf b}=b_1\cdots b_r, v^R.$$

For example, with $R=4,$ two such words are

$$v^4{\bf a}v^4{\bf b}v^8{\bf a}v^4{\bf a}v^{12}, v^4{\bf a}v^4{\bf ba}v^4{\bf a}v^{12}.$$
If the first and last factors of such a word are $v^R,$ and if two consecutive factors always include at least one $v^R,$ then we call it an {\it $R$-word}.  The first word in the example above is a 4-word, but the second is not, because of the consecutive factors {\bf b} and {\bf a}. In what follows, we will concern ourselves exclusively with $R$-words. The way in which we have defined the word $v$ ensures that the factorization of an $R$-word in the required form is unique.

 Let $m\geq 0,$ and $R> 2m.$  We'll define special factors in $R$-words that we call {\it $m$-neighborhoods}. One kind of $m$-neighborhood is a factor  of the form  $v^m{\bf a}v^m$ or $v^m{\bf b}v^m,$ where the ${\bf a}$ or ${\bf b}$ is one of the original factors used to build the word.  In addition, we say that the prefix $v^m$ and suffix $v^m$ are also $m$-neighborhoods.  So, for example, the 1-neighborhoods in the 4-word in the example above are indicated here by underlining:
 
 $$\underline{v}\cdot v^2\cdot\underline{v{\bf a}v}\cdot v^2\cdot\underline{v{\bf b}v}\cdot v^6\cdot\underline{v{\bf a}v}\cdot v^2\cdot\underline{v{\bf a}v}\cdot v^{10}\cdot\underline{v}.$$
 
  The condition $R> 2m$ ensures that $m$-neighborhoods are never directly adjacent, so that every position belongs to at most one $m$-neighborhood, and some of the $v$ factors are contained in no $m$-neighborhood.
 
 Consider two marked words $(w_1, i_1), (w_2,i_2)$ where $w_1, w_2$ are $R$-words. We say these marked words are $\equiv_0^m$-equivalent if $w_1(i_1)=w_2(i_2),$ and if either $i_1$ and $i_2$ are in the same position in identical $m$-neighborhoods, or if neither $i_1$ nor $i_2$ belongs to a neighborhood.   For instance, if $m=2,$ and $i_1$ is on the third position of a 2-neighborhood $v^2{\bf a}v^2$ in $w_1,$ then $i_2$ will be on the third position of a 2-neighborhood $v^2{\bf a}v^2$ of $w_2.$  If $m=0,$ then we only require $w_1(i_1)=w_2(i_2),$ so $\equiv_0^0$ equivalence is the same as $\equiv_0$-equivalence.
 
 We now play our game  in marked words $(w_1,i_1), (w_2,i_2),$  where $w_1, w_2$ are $R$-words. We add the rule that at the end of every round, the two marked words $(w_1,j_1)$ and $(w_2,j_2)$ are $\equiv_0^m$-equivalent.  If Player 2 has a winning strategy in the $k$-round game with this additional rule, we say that the starting words $(w_1,i_1)$ and $(w_2,i_2)$  are  $\equiv_k^m$-equivalent. Once again, the case $m=0$ corresponds to ordinary $\equiv_k$-equivalence.  
 
 We will call this stricter version of the game the {\it $m$-enhanced game}. As with the original game, we can define a version of the $m$-enhanced game for ordinary (that is, unmarked) $R$- words: In the first round, Player 1 places his pebble on a position in either of the words, and Player 2 responds so that the resulting marked words are $\equiv_0^m$-equivalent. Play then proceeds as described above for $k-1$ additional rounds.  We write $w_1\equiv^m_k w_2$ if Player 2 has a winning strategy in this $k$-round $m$-enhanced game.
 
 Let $S,T>0,$ and let $X_{S,T}$ denote the $S$-word
 
 $$(v^S{\bf a}v^S{\bf b}v^S)^T.$$
 
 We claim:
\begin{lemma}~\label{lemma.xst}
For each $m\geq 0,$ $k\geq 1,$ there exists $R$ such that if $S,T\geq R,$
 $$X_{S,T}X_{S,T}\equiv_k^m X_{S,T}{\bf a}X_{S,T}.$$
\end{lemma}
\begin{proof} 
  We prove this by induction on $k,$ first considering the case $k=1$ with arbitrary $m\geq 0.$ Choose $R> 2m,$ and $S,T\geq R.$ If Player 1 plays in either word inside one of the factors $X_{S,T},$ then Player 2 might try to simply mimic this move in the corresponding factor $X_{S,T}$ in the other word.  This works unless Player 1 moves near the center of one of the two words.  For example, if Player 1 moves in the final position of the first $X_{S,T}$ in $X_{S,T}X_{S,T},$ then this position is contained inside a factor $v$ and does not belong to any $m$-neighborhood, but the corresponding position in the other word belongs to a neighborhood of the form $v^m{\bf a}v^m,$ so the response is illegal. Of course we can solve this problem: $X_{S,T}$ itself contains factors $v$ that do not belong to any $m$-neighborhood (because $R> 2m$).  Conversely, if Player 1 moves anywhere in the central $m$-neighborhood $v^m{\bf a}v^m$ in $X_{S,T}{\bf a}X_{S,T}$ then Player 2 can find an identical neighborhood inside $X_{S,T}$ and reply there.
  
Now let  $k\geq 1$ and suppose that the Proposition is true for this fixed $k$ and all $m\geq 0.$ We will show the same holds for $k+1.$ Let $m\geq 0.$ Then by the inductive hypothesis, there exist $S,T$ such that
$$X_{S,T}X_{S,T}\equiv^{m+1}_k X_{S,T}{\bf a}X_{S,T}.$$
We will establish the proposition by showing
$$X_{S,T+1}X_{S,T+1}\equiv^{m}_{k+1}X_{S,T+1}{\bf a}X_{S,T+1}.$$
Observe that
$$X_{S,T+1}X_{S,T+1}=X_{S,1}(X_{S,T}X_{S,T})X_{S,1},$$
$$X_{S,T+1}{\bf a}X_{S,T+1}=X_{S,1}(X_{S,T}{\bf a}X_{S,T})X_{S,1}.$$
We will call the factors $X_{S,1}$ occurring in these two words the {\it peripheral factors}; the remaining letters make up the {\it central regions}. We prove the proposition by presenting a winning strategy for Player 2 in the $(k+1)$-round $m$-enhanced game in these two words. 

For the first $k$ rounds, Player 2's strategy is as follows:
\begin{itemize}
\item If Player 1 moves into either of the peripheral factors in either of the words, Player 2 responds at the corresponding position of the corresponding peripheral factor in the other word.
\item If Player 1 moves from a peripheral factor into the central region of one word, Player 2 treats this as the opening move in the $k$-round $(m+1)$-enhanced game in the central regions, and responds according to her winning strategy in this game.
\item If Player 1 moves from one position in the central region of a word to another position in the central region, Player 2 again responds according to her winning strategy in the $k$-round $(m+1)$-enhanced game in the central regions.
\end{itemize}

 We need to show that each move in this strategy is actually a legal move in the game---that in each case the sets of letters jumped by the two players are the same, and that the $m$-neighborhoods match up correctly.  It is trivial that the $m$-neighborhoods match up---in fact, the $(m+1)$-neighborhoods do, so we concentrate on  showing that the sets of jumped letters are the same.

This is clearly true for moves that remain within a single peripheral factor or move from one peripheral factor to another.  What about the situation where Player 1 moves from a peripheral factor to the central region?  We can suppose without loss of generality that the move begins in the $i^{th}$ position of the left peripheral factor in one word, and jumps right to the $j^{th}$ position of the central region of the the word. If $j\leq |v|,$ then this move is within an $m$-neighborhood of the central factor, namely the prefix of the central factor of the form $v^m,$ so Player's 2 reply will be at the identical position in the corresponding neighborhood, and thus at position $j$ of the central factor in the other word.  Thus the two moves jumped over precisely the same set of letters. If $j>|v|,$ then Player 1's move jumps over all the letters of $v.$  Since Player 2's response, owing to the condition on neighborhoods, cannot be within the first $|v|$ letters of the central factor of the other word, her move too must jump over all the letters of $v.$ The same argument applies to moves from the central region into a peripheral factor.

This shows that Player 2's strategy is successful for the first $k$ rounds, but we must also show that Player 2 can extend the winning strategy for one additional round. If after the first $k$ rounds, the two pebbles are in peripheral factors, then Player 1's next move is either within a single peripheral factor, from one peripheral factor to another, or from a peripheral factor into the central region.  In all these cases Player 2 responds exactly as she would have during the first $k$ rounds. As we argued above, this response is a legal move.   

If Player 1 moves from the central region to one of the peripheral factors, Player 2 will move to the same position in the corresponding peripheral factor in the other word.  Again, we argue as above that this is a legal move.

The crucial case is when Player 1's move is entirely within one of the central factors. Now we can no longer use the winning strategy in the game in the central factors to determine Player 2's response, because she has run out of moves. We may suppose, without loss of generality, that Player 1's move is toward the left. We consider whether or not the set of letters that Player 1 jumps consists of all the letters in $B.$  If it does, then Player 2 can just locate a matching neighborhood somewhere in the left peripheral factor in the other word and move there.  In the process, Player 2 also jumps to the left over all the letters of $B.$

What if Player 1 jumps over a proper subset of the letters of $B$?   Player 2 responds by moving the same distance in the same direction in the other word. Why does this work? After $k$ rounds, the pebbles in the two words must be on the same letters and in the same positions within matching $(m+1)$-neighborhoods, or outside of any $(m+1)$-neighborhood.  Thus we cannot have the situation where, for example, the pebble in one word is on a letter $b_i$ within a factor ${\bf b},$ and in the other word on the same letter $b_i$ within a factor $v.$  If Player 1 moves the pebble from a factor $v$ into a factor ${\bf b}$ or ${\bf a}$---let us say ${\bf b}$--- then the original position of the pebble was within a 1-neighborhood.  Therefore Player 2's pebble was in the corresponding position in a matching 1-neighborhood, so the response moves this pebble to ${\bf b}$ as well, and thus jumps over the same set of letters. (This is the only potentially problematic case.)  How can we ensure that the $m$-neighborhoods match up correctly after this move?  If  the move takes Player 1's pebble into an $m$-neighborhood, then the move stayed within a single $(m+1)$-neighborhood.  This means that the original position of Player 2's pebble was in the corresponding position of an identical $(m+1)$-neighborhood, so that after Player 2's response, the two pebbles are in matching $m$-neighborhoods.

\end{proof}

\subsection{Proof of Corollary~\ref{cor.strictinclusion}}

Let 
$L= (a(ab)^*b)^*.$
It is easy to construct a sentence of $FO[<]$ defining $L.$ We can give a more algebraic proof by working directly with the minimal automaton of $L$ and showing that its transition monoid is aperiodic.  This automaton has three states $\{0,1,2\}$ along with a dead state.  For each state $i=0,1,2$ we define $i\cdot a=i+1,$ $i\cdot b=i-1,$ where this transition is understood to lead to the dead state if $i+1=3$ or $i-1=-1.$ The transition monoid has a zero, which is the transition mapping every state to the dead state. It is then easy to check that every $m\in M(L)$ is either idempotent, or has $m^3=0,$ and thus $M(L)$ is aperiodic. 
We now show that $L$ cannot be expressed in \fotwobet.  By Theorem~\ref{thm.maintheorem}, if  $L$ is expressible then $M(L)\in {\bf M}_e{\bf DA}.$ As before, we will denote the image of a word $w$ in $M(L)$ by $\overline {w}.$ Easily, $e=\overline{ab}$ is an idempotent in $M(L),$ and both $m=\overline{(ab)a(ab)}$ and $f=\overline{(ab)a(ab)b(ab)}$ are elements of $e\cdot M(L)_e\cdot e$ with $f$ idempotent and $f\leq_{\cal J} m$ in $e\cdot M(L)_e\cdot e.$  So if $e\cdot M(L)_e\cdot e\in{\bf DA}$ we would have
$fmf=f.$  Now $f$ is the transition that maps state 0 to 0 and all other states to the dead state, and $m$ is the transition that maps 0 to 1 and all other states to the dead state.  Thus $fmf=0\neq f,$ so $M(L)\notin {\bf M}_e{\bf DA}.$

\subsection{Proof of sufficiency in Theorem~\ref{thm.maintheorem} for two-letter alphabets}

 We sketch the proof that languages over a two-letter alphabet whose syntactic monoids are in ${\bf M}_e{\bf DA}$ are definable in \fotwobet, which will complete the proof of the main theorem. This relies heavily on an algebraic theory of finite categories developed by Tilson~\cite{Tilson}.  The category ${\cal C}(\phi)$ and the congruence $\cong$ described below were first introduced by Straubing~\cite{Str-dd2} in the study of languages of dot-depth 2.

In this and the following sections, we assume that $A=\{a,b\},$ that $M\in{\bf M}_e{\bf DA}$ and that $\phi:A^*\to M$ is a homomorphism onto $M.$  Our ultimate goal is to show that $\phi^{-1}(m)$ is definable in \fotwobet\/ for all $m\in M.$  This implies that every language recognized by a monoid in $M\in{\bf M}_e{\bf DA}$ is definable in the logic, and thus that $L$ is definable if $M(L)\in{\bf M}_e{\bf DA}$.

\subsubsection{Background on finite categories}

A {\it category} ${\cal C}$ consists of a set $Obj({\cal C})$ of {\it objects}, and, for each $c,d\in Obj({\cal C}),$ a set $Arr(c,d)$ of {\it arrows} from $c$ to $d.$  Given  a pair of consecutive arrows 
$$c\xrightarrow{u} d, d\xrightarrow{v} e,$$
there is a product arrow
$$c\xrightarrow{uv} e.$$
The product of nonconsecutive arrows is not defined.  We will sometimes denote arrows in this way by showing their start and end objects, and sometimes just write them as $u,$ $v,$ $uv,$ {\it etc.}. For each $c\in Obj({\cal C})$ there is an arrow $c\xrightarrow{1_c} c$ that is a right identity for arrows ending at $c$ and a left identity for arrows starting at $c.$  This is the traditional definition of a category, but we are not doing `category theory' in the traditional sense:  We will only consider categories in which both the object set and every arrow set is finite, and essentially treat categories as generalized monoids.  Note that a monoid is the same thing as a category with a single object.  Furthermore, for each object $c$ of a category ${\cal C},$ the set $Arr(c,c)$ of arrows that both begin and end at $c$ is a finite monoid, called the {\it base monoid} at $c.$

There is a preorder on categories called {\it division} and denoted $\prec.$  The ingredients of a division ${\cal C}\prec {\cal D}$ are an object map $\tau:Obj({\cal C})\to Obj({\cal D}),$ and for each arrow $u\in Arr(c,d),$ a set of arrows in $Arr(\tau(c),\tau(d))$ that are said to {\it cover} $u.$ The covering relation is required to satisfy a number of properties: 
\begin{enumerate}
\renewcommand\labelenumi{(\roman{enumi})}
\item An arrow in $Arr(\tau(c),\tau(d))$ can cover at most one arrow in $Arr(c,d).$  (It might  cover several different arrows with different starting and ending objects.) 
\item The covering relation is multiplicative:  If $\hat u$ covers $u$ and $\hat v$ covers $v,$ then the product arrow $\hat u\hat v$ covers $uv.$
\item For each $c\in Obj({\cal C}),$ $1_c$ is covered by $1_{\tau(c)}.$
\end{enumerate}

An important special case occurs when the category ${\cal D}$ is a monoid, and an even more special case when ${\cal C}$ and ${\cal D}$ are both monoids.  In this instance, category division reduces to monoid division:  ${\cal C}$ is a quotient of a submonoid of ${\cal D}.$

We will employ a special property of the monoid variety {\bf DA} with respect to categories:  It is {\it local}.  This means that if every base monoid of a category ${\cal C}$ divides a monoid in ${\bf DA},$ then ${\cal C}$ itself divides a monoid in {\bf DA} (see Almeida~\cite{Almeida}, Place and Segoufin~\cite{PS}).

\subsubsection{The category ${\cal C}(\phi)$} 

 The objects of  ${\cal C}(\phi)$ are pairs $(e,X),$ where $e\in M$ is idempotent, $X\subseteq A,$ and $e=\phi(w)$ for some word $w\in A^*$ with $\alpha(w)=X.$  (We write $\alpha(w)$ for the set of letters of $w.$)
An arrow  from $(e,X)$ to $(f,Y)$ is represented by a triple

$$(e,X)\stackrel{w}{\longrightarrow}(f,Y),$$

where $w\in A^*$ and $Y=X\cup\alpha(w).$

Two such arrows
$$(e,X)\xrightarrow{w_1,w_2}(f,Y)$$
are identified if 
$$e\cdot\phi(w_1)\cdot f=e\cdot\phi(w_2)\cdot f.$$
(That is, an arrow is actually an equivalence class modulo this identification.)

We multiply arrows by the rule
$$(e,X)\stackrel{w_1}{\longrightarrow}(f,Y)\stackrel{w_2}{\longrightarrow}(g,Z)$$ is equal to $$(e,X)\stackrel{w_1uw_2}{\longrightarrow}(g,Z),$$
where $u\in A^*,$ $\alpha(u)=Y,$ $\phi(u)=f.$

It is straightforward to verify that this multiplication is well-defined and associative, and that
$$(e,X)\xrightarrow{u}(e,X),$$
where $\alpha(u)=X$ and $\phi(u)=e,$ represents the identity arrow $1_{(e,X)}.$  Thus ${\cal C}(\phi)$ is indeed a category.

The base monoid at $(e,X)$ in ${\cal C}(\phi)$ is in $eM_ee.$ Since $M\in{\bf M}_e{\bf DA},$ the result about locality of {\bf DA} cited above implies that ${\cal C}(\phi)$ divides a monoid in {\bf DA}.

\subsubsection{A congruence on $\{a,b\}^*$}

We describe an equivalence relation on $\{a,b\}^*.$   Since $M\in{\bf M}_e{\bf DA},$ $M$ is aperiodic, and therefore if $T>|M|,$ $m^T=m^{T+1}$ for all $m\in M.$

Given a word $w\in\{a,b\}^*$ we factor it into maximal factors each of which consists entirely of $a$'s or entirely of $b$'s. We call these factors {\it sub-blocks}.  For example,
$$w=a^{k_r}b^{\ell_{r}}\cdots a^{k_1}b^{\ell_1}.$$
In this example the word begins with a sub-block of $a$'s and ends with a sub-block of $b$'s, but of course any of the four possibilities can occur.  We call  the factors $a^{k_i}b^{\ell_i}$ the {\it blocks} of the factorization. In case the number of sub-blocks is odd, the leftmost block will be incomplete, in the sense that it contains only $a$'s or only $b$'s. Observe that we number the blocks from right to left, and call the rightmost block the first block.

Let $w,w'\in \{a,b\}^*.$  We define $w\cong w'$ if the following conditions hold:

\begin{itemize}
\item The number of blocks in $w$ and the number of blocks in $w'$ are equivalent threshold $T.$
\item The rightmost letters of $w,w'$ are the same.  
\item Let $S$ be the number of blocks in $w.$  Let $1\leq i\leq \min(S,T),$ and let $a^{k_i}b^{\ell_i},$ $a^{k'_i}b^{\ell'_i}$ be the $i^{th}$ blocks of $w,w',$ respectively.  Then $k_i$ and $k'_i$ are equivalent threshold $T,$ and $\ell_i$ and $\ell'_i$ are equivalent threshold $T.$  (This assumes that the rightmost letter of both words is $b,$ so that blocks have the form $a^kb^{\ell},$ but the analogous definition is made if the rightmost letter is $a$.)
\end{itemize}

We have the following easy lemma.

\begin{lemma}\label{lemma.congcong} The equivalence relation $\cong$ is a congruence of finite index on $\{a,b\}^*.$
\end{lemma} 

We let $K$ denote the quotient monoid $A^*/\cong$ and $\psi:A^*\to K$ the projection homomorphism onto this quotient.

\subsubsection{The derived category}

We define another category $D$, called the {\it derived category}, that depends on the homomorphisms $\phi:A^*\to M,$ $\psi:A^*\to K.$  (In Tilson's formulation, this is called the derived category of the {\it relational morphism} $\psi\phi^{-1}:M\to K.$) 

The objects of $D$ are elements of $K.$  The arrows are represented by triples
$$k\stackrel{v}{\longrightarrow}k\cdot\psi(v),$$
where $k\in K$ and $v\in A^*.$  Two such triples
$$k\xrightarrow{v_1,v_2}k\cdot\psi(v_1)=k\cdot\psi(v_2)$$
are identified if for all words $u$ such that $\psi(u)=k,$ $\phi(uv_1)=\phi(uv_2).$
Triples are composed in the obvious way:  the product
$$k\xrightarrow{v_1}k\cdot\psi(v_1)\xrightarrow{v_2}k\psi(v_1v_2)$$
is represented by the triple
$$k\xrightarrow{v_1v_2}k\cdot\psi(v_1v_2).$$
It is straightforward to verify that this composition law is associative, respects the equivalence of arrows. and that $k\stackrel{1}{\longrightarrow} k$ is a left and right identity at the object $k.$  So this is indeed a category.

The crucial property of $D$ that we will use is: 

\begin{lemma}\label{lemma.derivedcategory}
$D$ divides a monoid in {\bf DA}.
\end{lemma}

The proof of this lemma is given in Straubing~\cite{Str-dd2}, where both the category ${\cal C}(\phi)$ and the congruence $\cong$ are introduced.

\subsubsection{Translation into logic}\label{section.translation}

 Let $w=a_1\cdots a_k\in A^*.$  Consider the path traced out in the derived category $D$ by this word, starting at the object $1=\psi(1)$: 
$$1\xrightarrow{a_1}\psi(a_1)\xrightarrow{a_2}\psi(a_1 a_2)\cdots \xrightarrow{a_k}\psi(a_1\cdots a_{k}).$$
We view this sequence of arrows as a word over the alphabet $B$ of all arrows in $D.$  This defines a length-preserving function (but not a homomorphism) $\xi:A^*\to B^*.$  By Lemma~\ref{lemma.derivedcategory}, $D$ divides a monoid $N\in {\bf DA},$ so we can map each arrow $b\in B$ to an element $\theta(b)$ that covers it.

Let $w,w'\in A^*.$ We claim that if both $\psi(w)=\psi(w'),$ and $\theta\xi(w)=\theta\xi(w'),$ then $\phi(w)=\phi(w').$ This is because the product 
$(1,w,\psi(w))$ of the sequence of arrows $\xi(w)$ is covered by $\theta\xi(w),$ which consequently covers the product $(1,w',\psi(w'))$ of the arrows in $\xi(w').$  As these products have the same start and end object, they are accordingly equal, which implies $\phi(w)=\phi(w').$ 

Thus if $m\in M,$ the language $\phi^{-1}(m)$  is a finite union of sets of the form
$$\{w\in\Sigma^*:\psi(w)=k,\theta\xi(w)=n\},$$
for $n\in N$ and $k\in K.$  We will show that each of these sets of words is definable in \fotwothr, and thus, by Theorem~\ref{thm.invequalsthr}, by a sentence of \fotwobet.
To express $\psi(w)=k,$ we need to be able to describe a $\cong$-class in \fotwothr. A single example will illustrate the general procedure.  Suppose that $T=3$ and we want to write a sentence satisfied by a word $w$ if and only if $w$ is congruent to
$$a^4b^2a^5ba^2b^4a.$$
This word has four blocks, and the leftmost block is incomplete, but only the three rightmost blocks matter for the congruence class. To describe the class, we first say that we jump left from the right-hand end of of word, over a single $a$, to the rightmost $b$:
$$\exists x(b(x)\wedge\forall y(y>x\rightarrow (a(y)\wedge\neg a(x,y)\wedge\neg b(x,y))).$$
We now jump from the rightmost $b$ over no $a$'s and at least one $b$ to the leftmost $b$ of the first block.  Observe that the condition is `at least one $b$,' because we want to express that the size of this sub-block is at least 3.  Thus our sentence has now grown to
$$\exists x(b(x)\wedge\forall y(y>x\rightarrow (a(y)\wedge\neg a(x,y)\wedge\neg b(x,y)))\wedge\gamma(x)).$$ 
where $\gamma(x)$ is
$$\exists y(y<x\wedge (b,1)(y,x)\wedge\neg a(y,x)).$$
We continue in this manner, now jumping to the rightmost $b$ of the next block, and then the leftmost $b$ of the same sub-block.  Note that we grow the sentence from the outside in, at each step adding another existential quantifier and atomic formulas of \fotwothr.  When we reach the leftmost letter of the third block we will have to add that either we are on the leftmost letter of the word, or that the letter immediately to the left is $a.$  The end result is a sentence of \fotwothr   defining a congruence class of $\cong,$ and we can treat every class similarly. %

How do we express that $\theta\xi(w)=n$?  Since $\theta$ is a morphism into $N\in {\bf DA},$ the set of $v\in B^*$ such that $\theta(v)=n$ is expressed by a sentence of $FO^2[<]$ over the alphabet $B.$  To convert this into a sentence over $A,$ we have to translate each atomic formula $(k,a,k\cdot\psi(a))(x)$ into a formula of \fotwothr\/ over $A.$  We write this as 
$a(x)\wedge \alpha(x),$
where $\alpha$ says that the word consisting of letters to the left of $x$ belongs to the $\cong$-class $k.$  We then proceed to write $\alpha$ as formula over \fotwothr, following the same procedure we used above to obtain sentences describing $\cong$-classes. In fact, this formula is simpler, in terms of alternation depth, than the one we produced earlier: Using the same $\cong$-class as in our example, our formula construction will begin with
$$\exists y(y<x\wedge (a,1)(y,x)\wedge \neg(a,2)(y,x)\wedge\neg(b,1)(y,x)).$$
The result is a sentence of \fotwothr defining $\phi^{-1}(m),$ as required.

\subsection{Proof of Theorem~\ref{thm.alternationdepth}}

We first show that \fotwobet\/ contains languages of arbitrarily large alternation depth.

Consider an alphabet consisting of the symbols
$$0,1,\vee_1,\wedge_2,\vee_3,\wedge_4,\cdots.$$
We define a sequence of languages by regular expressions as follows:
$$C_1=\vee_1(0+1)^+$$
$$T_1=\vee_1(0+1)^*1(0+1)^*.$$
For even $m>1,$
$$C_m=\wedge_mC_{m-1}^+$$
$$T_m=\wedge_mT_{m-1}^+.$$
For odd $m>1,$
$$C_m=\vee_mC_{m-1}^+$$
$$T_m=\vee_mC_{m-1}^*T_{m-1}C_{m-1}^*.$$

Observe that $C_m$ and $T_m$ are languages over a finite alphabet of $m+2$ letters.  $C_m$ denotes the set of prefix encodings of depth $m$ boolean circuits with 0's and 1's at the inputs.  In these circuits the input layer of 0's and 1's is followed by a layer of unbounded fan-in OR gates, then alternating layers of unbounded fan-in AND and OR gates (strictly speaking, these circuits are trees of AND and OR gates).  $T_m$ denotes the set of encodings of those circuits in $C_m$ that evaluate to {\it True.}

 We obtain a sentence  defining $C_1$ by saying that the first symbol is $\vee_1$ and every symbol after this is either 0 or 1.  
$$\exists x( \vee_1(x) \wedge \forall y(x\leq y \wedge x<y\rightarrow (1(y) \vee 0(y)))).$$
We obtain a sentence for $T_1$ by taking the conjunction of this sentence with $\exists x 1(x).$

For the inductive step, we suppose that we have a sentence defining $T_k$ for $k\geq 1.$  Let's suppose first that $k$ is odd.  Thus $T_k$ is a union of $\equiv_r$-classes for some $r,$ which we assume to be at least 2. Consequently, whenever $v\in T_k$ and $v'\notin T_k,$ Player 1 has a winning strategy in the $r$-round game in $v$ and $v'.$  The proof now proceeds by showing that whenever $w\in T_{k+1}$ and $w'\notin T_{k+1},$ Player 1 has a winning strategy in the $(r+1)$-round game in these two words.  This implies that the $\equiv_r$-class of $w$ is contained within $T_{k+1},$ and thus $T_{k+1}$ is a union of such classes, and hence definable by a sentence in our logic.

 The argument for the case where $k$ is even, and for the classes $C_k,$ is identical.  Observe that since $T_1$ is defined by a formula of quantifier depth 2, we can take $r=k+1.$

It remains to prove the claim about unbounded alternation depth.  It is possible to give an elementary proof  of this using games. However, by deploying some more sophisticated results from circuit complexity, we can quickly see that the claim is true. Let us suppose that we have a language $L\subseteq\{0,1\}^*$ recognized by a constant-depth polynomial-size family of unbounded fan-in boolean circuits; that is, $L$ belongs to the circuit complexity class $AC^0.$  We can encode the pair consisting of a word $w$ of length $n$ and the circuit for length $n$ inputs by a word $p(w) \in C_n.$ We now  have $w\in L$ if and only if $p(w)\in T_n.$  

Now if the alternation depth of all the $T_n$ is bounded above by some fixed integer $d,$  then we can recognize every $T_n$ by a polynomial-size family of circuits of depth $d.$  We can use this to obtain a polynomial family of circuits of depth $d$ recognizing $L.$  This contradicts the fact (see Sipser~\cite{Sipser}) that the required circuit depth of languages in $AC^0$  is unbounded.

To prove the claim about bounded alternation depth in the two-letter case, we re-examine the construction of the formulas in Section~\ref{section.translation}.   The formula defining the condition $\theta\xi(w)=n$ is constructed by writing a sentence of $FO^2[<]$ over a base that includes the atomic formulas $(k,a,k\cdot\psi(a))(x)$.  The formula $\alpha(x)$ used to express this atomic formula is in $\Sigma_2[<].$  The $FO^2[<]$ sentence itself can be replaced (Th\'erien and Wilke~\cite{TW}) by an equivalent $\Sigma_2[<]$ sentence or a $\Pi_2$ sentence, so the result is a $\Pi_3[<]$ sentence defining $\theta\xi(w)=n.$  The formula sentence defining $\psi(w)=k$ can itself be replaced by a $\Sigma_2[<]$ or $\Pi_2[<]$ sentence, because $K$ is ${\cal L}$-trivial and thus itself in {\bf DA}. 
\section*{Acknowledgments}
Much of this research was carried out while the various authors were guests of the Tata Institute for Fundamental Research in Mumbai, the Chennai Mathematical Institute and the Institute of Mathematical Sciences in Chennai, and the University of Montreal, and participated in the Dagstuhl Seminar `Circuits, Logic and Games' in September, 2015. 

\bibliographystyle{abbrvnat}

\end{document}